\documentclass{sig-alternate}
\usepackage{graphicx}
\usepackage{bussproofs}
\usepackage{multirow}
\usepackage{flushend}
\usepackage{amsmath}
\usepackage{algorithm}
\usepackage[noend]{algpseudocode}
\usepackage{hyperref}
\usepackage{listings}
\begin{document}


\title{A Formal Graph Model for RDF and Its Implementation}
\numberofauthors{4} 
\author{
\alignauthor{
Vinh Nguyen\\
 \affaddr{Kno.e.sis Center}\\
 \affaddr{Wright State University}\\
 \affaddr{Ohio, USA}\\
 \email{vinh@knoesis.org}
}
\alignauthor{
Jyoti Leeka\\
 \affaddr{IIIT Delhi}\\
 \affaddr{India}\\
 \email{jyotil@iiitd.ac.in}
}\alignauthor{
Olivier Bodenreider\\
 \affaddr{National Library of Medicine}\\
 \affaddr{National Institute of Health}\\
  \affaddr{Maryland, USA}\\
 \email{olivier@nlm.nih.gov}
}
\and
\alignauthor{
Amit Sheth\\
 \affaddr{Kno.e.sis Center}\\
 \affaddr{Wright State University}\\
 \affaddr{Ohio, USA}\\
 \email{amit@knoesis.org}
}
}
\maketitle

\newtheorem{theorem}{Theorem}[section]
\newtheorem{lemma}[theorem]{Lemma}
\newtheorem{proposition}[theorem]{Proposition}
\newtheorem{corollary}[theorem]{Corollary}

\begin{abstract}
Formalizing an RDF abstract graph model to be compatible with the RDF formal semantics has remained one of the foundational problems in the Semantic Web. In this paper, we propose a new formal graph model for RDF datasets. This model allows us to express the current model-theoretic semantics in the form of a graph. We also propose the concepts of resource path and triple path as well as an algorithm for traversing the new graph. We demonstrate the feasibility of this graph model through two implementations: one is a new graph engine called GraphKE, and the other is extended from RDF-3X to show that existing systems can also benefit from this model. In order to evaluate the empirical aspect of our graph model, we choose the shortest path algorithm and implement it in the GraphKE and the RDF-3X. Our experiments on both engines for finding the shortest paths in the YAGO2S-SP dataset give decent performance in terms of execution time. The empirical results show that our graph model with well-defined semantics can be effectively implemented.

\end{abstract}

\section{Introduction}








As the adoption of Semantic Web grows, a large number of datasets are generated and made publicly available, for example, as in the Linked Open Data. Due to the graph nature of these datasets, many graph-based algorithms for RDF datasets have been developed such as query processing \cite{hartig2007sparql,Zeng2013}, graph matching \cite{brocheler2009dogma,Cheng2008}, semantic associations \cite{Anyanwu:2005:SRC:1060745.1060766,Anyanwu:2003:9EQ:775152.775249}, path computing \cite{heim2009relfinder,przyjaciel2012rdfpath}, and centrality \cite{Cheng2011,Zhang2007}. These existing graph algorithms employ the abstract graph model called Node-Labeled Arc-Node (NLAN), currently recommended by W3C \cite{cyganiak2014}. 

In this NLAN graph, the subject and object of a triple are mapped to two nodes of a graph, and the predicate is mapped to a directed labeled arc connecting the subject and the object nodes of that graph. Although these straightforward mappings work for simple triples at the instance level, challenges may arise where predicates also play the role of the subject in other triples as in Fig. \ref{disconnected-subgraphs}. 

Particularly, in the first triple, the predicate is mapped to a labeled arc connecting the two subject/object nodes of the subgraph G1 representing the first triple. When the predicate from the first triple becomes the subject or the object of the second triple, it is mapped to another node of the subgraph G2 representing the second triple. As a result, a predicate is mapped to both a labeled arc in the subgraph G1 and a node within the subgraph G2. Despite the fact that the predicate is shared between the two triples, the resulting graph is comprised of two disconnected subgraphs G1 and G2, making it impossible to traverse between the subgraphs. The \verb|Object 2| is unreachable from \verb|Subject 1|. In addition, it is unfortunate for the NLAN graph model that such a scenario is common as RDF syntax allows us to make assertions about any resources, including predicates. Here we present two of such scenarios: schema triples and singleton property triples. We will take the singleton property triples forward to motivate our work.

\begin{figure}[h!]
 \centering
 \includegraphics[width=0.30\textwidth]{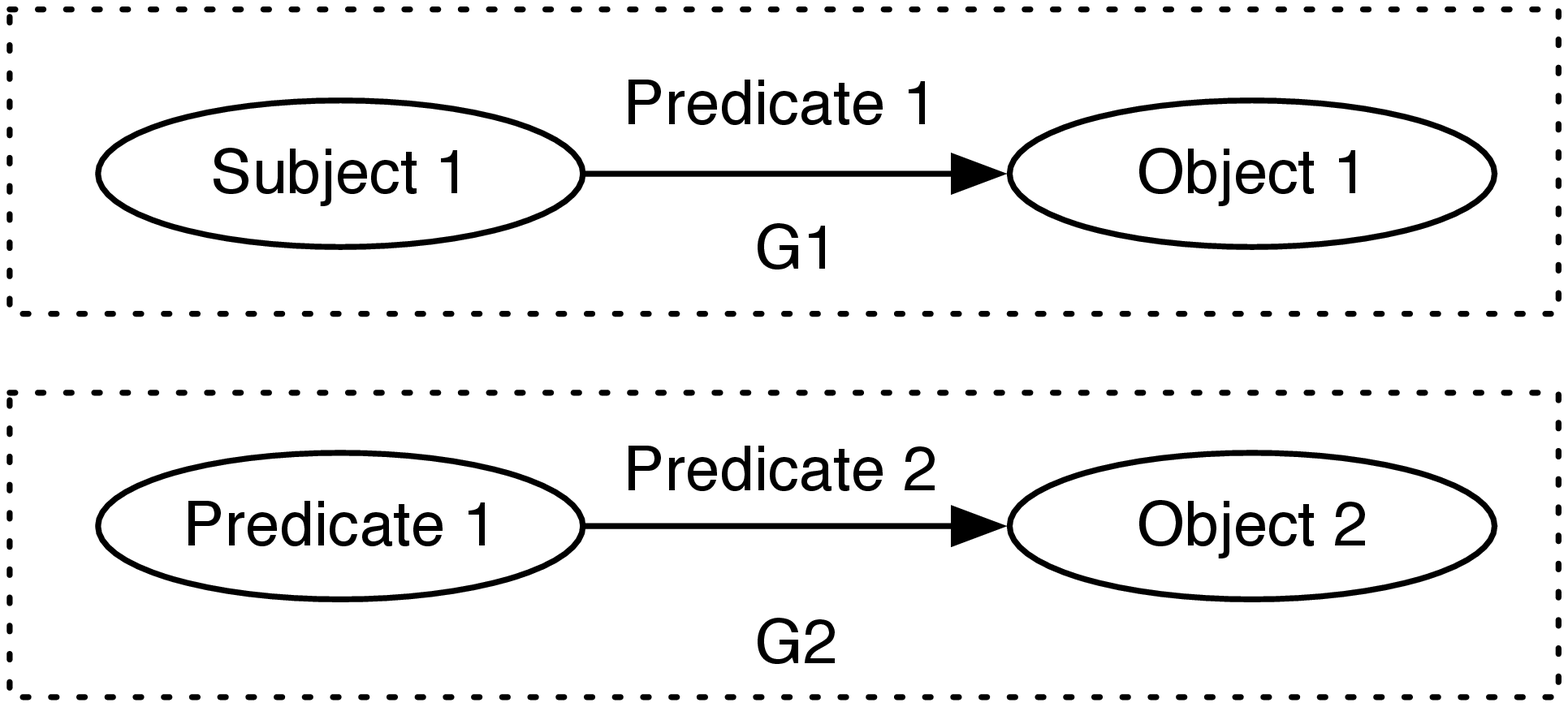}
 \caption{Subgraphs G1 and G2 are disconnected in the Node-LabeledArc-Node diagram (NLAN). \label{disconnected-subgraphs}}
\end{figure}

\textbf{Schema triples}. One common scenario where a predicate becomes the subject or object of another triple is through the RDF schema triples. A schema triple may describe a property using many other different properties such as \verb|rdfs:subPropertyOf|, \verb|rdf:type|, \verb|rdfs:domain|, or \verb|rdfs:range|. For example, the predicates \verb|hasFamilyName| and \verb|hasGivenName| are asserted as sub properties of the predicate \verb|rdfs:label|, and hence become the subjects in these triples: \\
\indent \indent \verb|hasFamilyName rdfs:subPropertyOf rdfs:label .|\\
\indent \indent \verb|hasGivenName rdfs:subPropertyOf rdfs:label .|\\
The Yago2S-SP dataset \cite{Nguyen:2014:DLR:2566486.2567973} contains 938 properties like these. We can also find schema triples in many other common datasets, such as DBPedia, etc.

\textbf{Singleton Property triples.} Another scenario comes from the RDF datasets created by a recent work \cite{Nguyen:2014:DLR:2566486.2567973} where singleton properties are used to describe RDF statements. A simple approach to make assertions about an RDF statement is through reification. Reification creates an instance of class \verb|rdf:Statement| to represent a statement and asserts metadata about this statement through this instance. Although this approach is intuitive, reifying one statement requires at least four triples, making it less attractive. Instead of reifying a statement, Nguyen et al. \cite{Nguyen:2014:DLR:2566486.2567973} propose a new approach called the singleton property. The singleton property approach was shown to be more efficient than reification. Each singleton property uniquely representing a statement can be annotated with different kinds of metadata describing that statement such as provenance, time, and location. Therefore, the singleton properties also become subjects and objects of other meta triples. Meta triples are the triples that describe other triples through the use of singleton properties. While singleton properties enrich RDF datasets with different kinds of metadata for the triples they represent, the metadata subgraph is disconnected from the triple they describe. Traversing the NLAN graphs created by this approach also becomes a challenge for the RDF graph model. 

Being unable to traverse among triples in the scenarios described above due to the limited connectivity of the NLAN graph is a critical issue, because it limits the capability of answering reachability and shortest path queries on RDF datasets. A reachability query verifies if a path exists between any two nodes in a graph. A shortest path query returns the shortest distance in terms of the number of edges between any two nodes in a graph if the path exists. In the NLAN graph model, both query types require traversing the graph from subject node to object node of a triple. Traversing from one node in subgraph G1 to another node in subgraph G2 in this way will not find any result because they do not share any node in common. The only resource these two triples share is the predicate of the first triple, which is also the subject of the second triple. Therefore, the ability to connect the triples represented in G1 with the corresponding triples in G2 (schema triples or meta triples represented through the singleton property) is desirable. Such connectivity strengthens the robustness of the graph model. 

Next, we will take one sample set of RDF triples represented by the singleton property approach as our motivating example. We will analyze the limitations of the existing work through this motivating example and demonstrate our approach to address the challenges.

\subsection{Motivating Example}
\label{motivating}

Consider the set of triples in Table \ref{tbl-motivating-example}, which will be used as a running example throughout the paper. In this example, we represent the facts that Bill Clinton is succeeded by George W. Bush as the president of the United States, and is succeeded by Frank White as the Governor of Arkansas. If we represent the facts as follows, these relationships are unclear and incomplete because the graph does not represent the political position context in which Bill Clinton is succeeded.\\
\indent \indent \verb|BillClinton hasSuccessor GeorgeWBush .|\\
\indent \indent \verb|BillClinton hasSuccessor FrankWhite .|

Instead, we propose to use the singleton property approach to represent the facts as shown in Table \ref{tbl-motivating-example}. Indeed, our example is motivated by the facts from the Yago2S-SP dataset, which will be used in the evaluation described in Section \ref{evaluation}. We create our example in the way that makes it intuitive to readers, and for readability, we eliminate all the URI prefixes. For each political position of Bill Clinton, we create a singleton property to capture the information related to that context. Particularly, we create two singleton properties \verb|holdsPos#1| and \verb|holdsPos#2|. Then we attach the predicate \verb|hasSuccessor| to the two singleton properties to represent the 3-ary relationship among the politician, the position and the successor.


We will analyze how such a set of triples can be represented in the form of a graph by two existing approaches: the NLAN diagram \cite{cyganiak2014} and the bipartite (BI) model \cite{hayes2004bipartite}.

\begin{table}
\centering
\caption{Example triples using singleton properties to represent the facts about the American politician Bill Clinton and his successors
\label{tbl-motivating-example}}
\begin{tabular}{cllll } \noalign{\smallskip}\hline \noalign{\smallskip}
 \textbf{$N_o$} & \textbf{Subject} & \textbf{Predicate} & \textbf{Object} \\ \hline \noalign{\smallskip}
 $T_1$&BillClinton & holdsPos\#1 & U.S.President \\ 
 $T_2$&holdsPos\#1 & singletonPropOf & holdsPos \\ 
 $T_3$&holdsPos\#1 & hasSuccessor &  GeorgeWBush \\
 $T_4$&BillClinton & holdsPos\#2 & ArkansasGovernor \\ 
 $T_5$&holdsPos\#2 & singletonPropOf & holdsPos \\ 
 $T_6$&holdsPos\#2 & hasSuccessor & FrankWhite \\
 \hline \noalign{\smallskip}
 \end{tabular}
\end{table}


A triple may be represented as a NLAN diagram as explained in the W3C RDF 1.1 Concepts and Syntax \cite{cyganiak2014}. The set of triples in Table \ref{tbl-motivating-example} are represented as a NLAN diagram in Fig. \ref{nlan_example}. We observe two problems with this modeling when dealing with such a set of triples.

\begin{figure}[h!]
 \centering
 \includegraphics[width=0.45\textwidth]{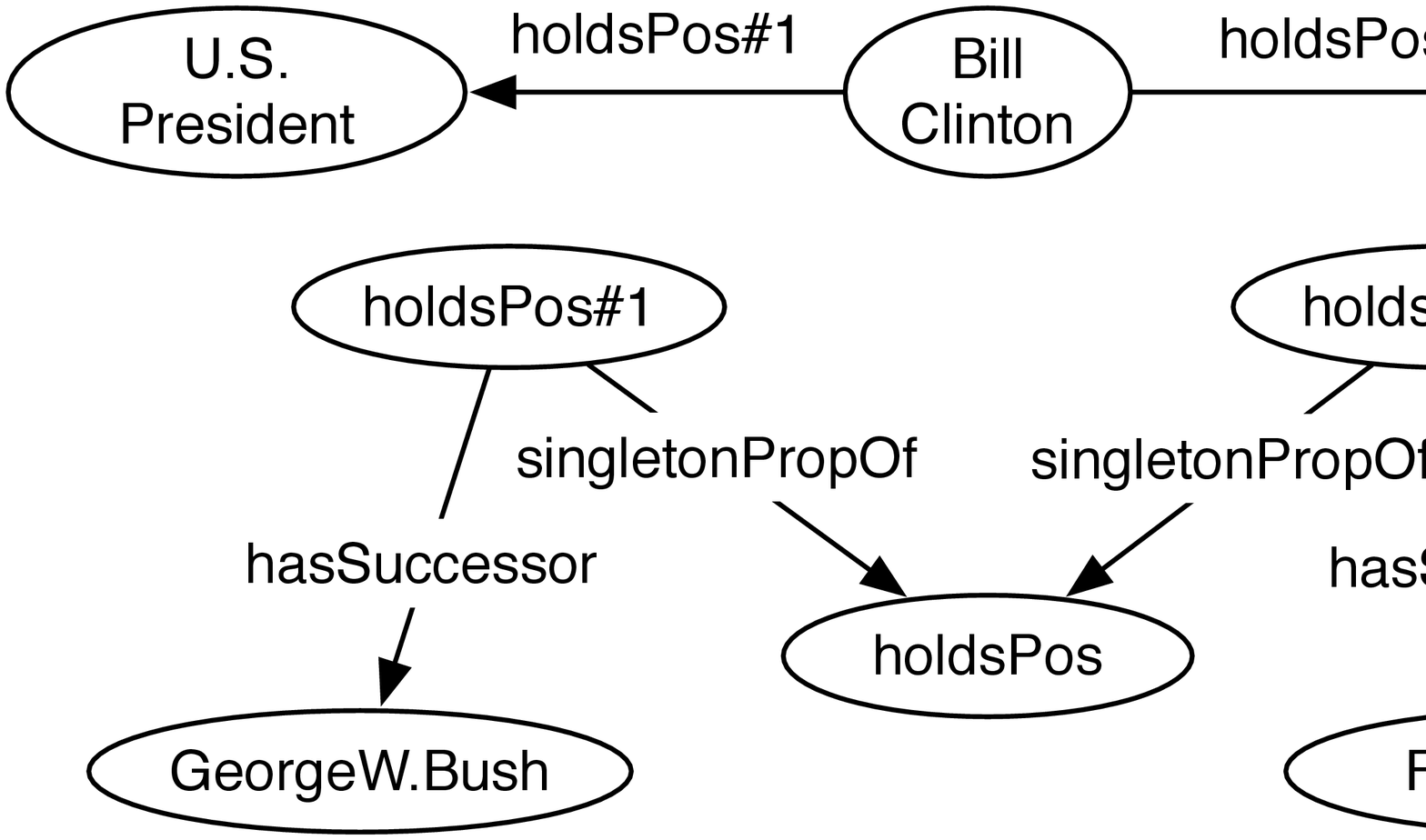}
 \caption{Node-LabeledArc-Node diagram (NLAN) for the example in Table \ref{tbl-motivating-example}. \label{nlan_example}}
\end{figure}

First and more importantly, the NLAN resulting graph is not precisely a mathematical graph. In Fig. \ref{nlan_example}, the RDF term \verb|holdsPos#1| is mapped to the predicate arc of triple $T_1$, and to the subject node in $T_2$, $T_3$. Mapping the same RDF term \verb|holdsPos#1| to more than one different mathematical object makes the object being referred to ambiguous. No mapping function satisfies this because a function must map any source object to only one target. For this reason, this NLAN model cannot be compatible with the formal semantics that is defined by several mapping functions.

The second serious problem is that the disconnectedness between subgraphs limits the possibility to traverse between the triples that are connected through the predicates. In the example at hand, the NLAN resulting graph is comprised of two disconnected subgraphs G1 (formed by two triples $T_1$ and $T_4$) and G2 (formed by four triples $T_2$, $T_3$, $T_5$ and $T_6$). Although the two predicates \verb|holdsPos#1| and \verb|holdsPos#2| are shared between the triples from two subgraphs, no connectivity between two subgraphs can be found in Fig. \ref{nlan_example}. As a consequence, although Bill Clinton is succeeded by George W. Bush, no path can be found from Bill Clinton to George W. Bush. From this perspective, the bipartite model provides better graph connectivity than the NLAN model.


The bipartite (BI) model proposed by Hayes \cite{hayes2004bipartite} does not encounter these problems of the NLAN model. It represents all subjects, predicates, and objects as nodes. It creates an auxiliary node and links this node to the subject, the predicate, and the object of this triple via three additional arcs. However, the cost incurred in this approach with one extra node and three extra arcs for every triple is too high.




\subsection{Our approach}

In this paper, we propose a new formal graph model which represents RDF triples in a Labeled Directed Multigraph with Triple Nodes (LDM-3N). Similar to the bipartite model, all subjects, predicates, and objects are mapped to nodes of a LDM-3N graph. This model, however, differs from the bipartite model in that it adds one pair of directed edges (subject-predicate, predicate-object) to directly connect three nodes of the same triple. Furthermore, this approach differs from the NLAN diagram in that the predicates are mapped to nodes instead of labeled edges. 

As traversing the NLAN graph starts from a node and explores its adjacent nodes, predicates as arcs never get explored. In the LDM-3N model, predicates are mapped to nodes in the graph. As they are adjacent to the subjects, they are explored after the subject. The objects now become adjacent to the predicates and get explored after the predicates. Traversing the LDM-3N graph this way starting from the node \verb|Subject 1| will reach the node \verb|Object 2| after 3 hops, as shown in Fig. \ref{connected-subgraphs}.

\begin{figure}[h!]
 \centering
 \includegraphics[width=0.45\textwidth]{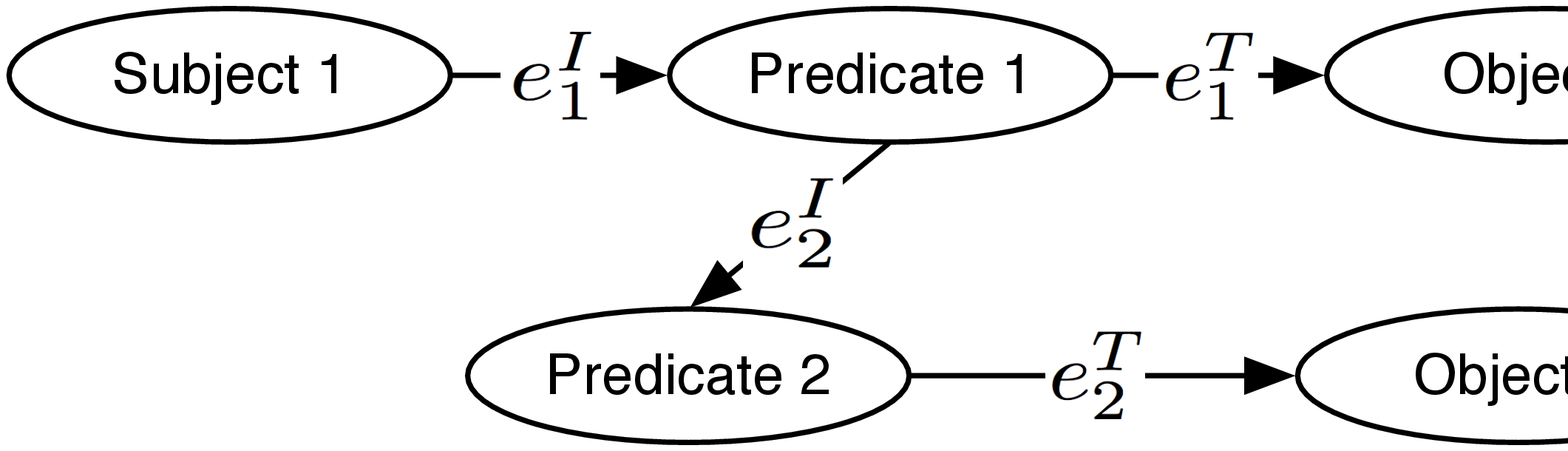}
 \caption{Subgraphs G1 and G2 from Fig. \ref{disconnected-subgraphs} are now connected in a Labeled Directed Multigraph with Triple Nodes (LDM-3N). \label{connected-subgraphs}}
\end{figure}

The LDM-3N graph representation of the example from Table \ref{tbl-motivating-example} is illustrated in Fig. \ref{fg-motivating-example}. For the set of six triples in Table \ref{tbl-motivating-example}, we map all subjects, predicates, and objects to the set of ten nodes. We link every three nodes representing the same triple by one pair of directed edges. For instance, the three nodes \{\verb|BillClinton|, \verb|holdsPos#1|, \verb|U.S.President|\} are connected by the pair of edges ($e^I_1$, $e^T_1$) to form the triple $T_1$.

\begin{figure}[h!]
 \centering
 \includegraphics[width=0.48\textwidth]{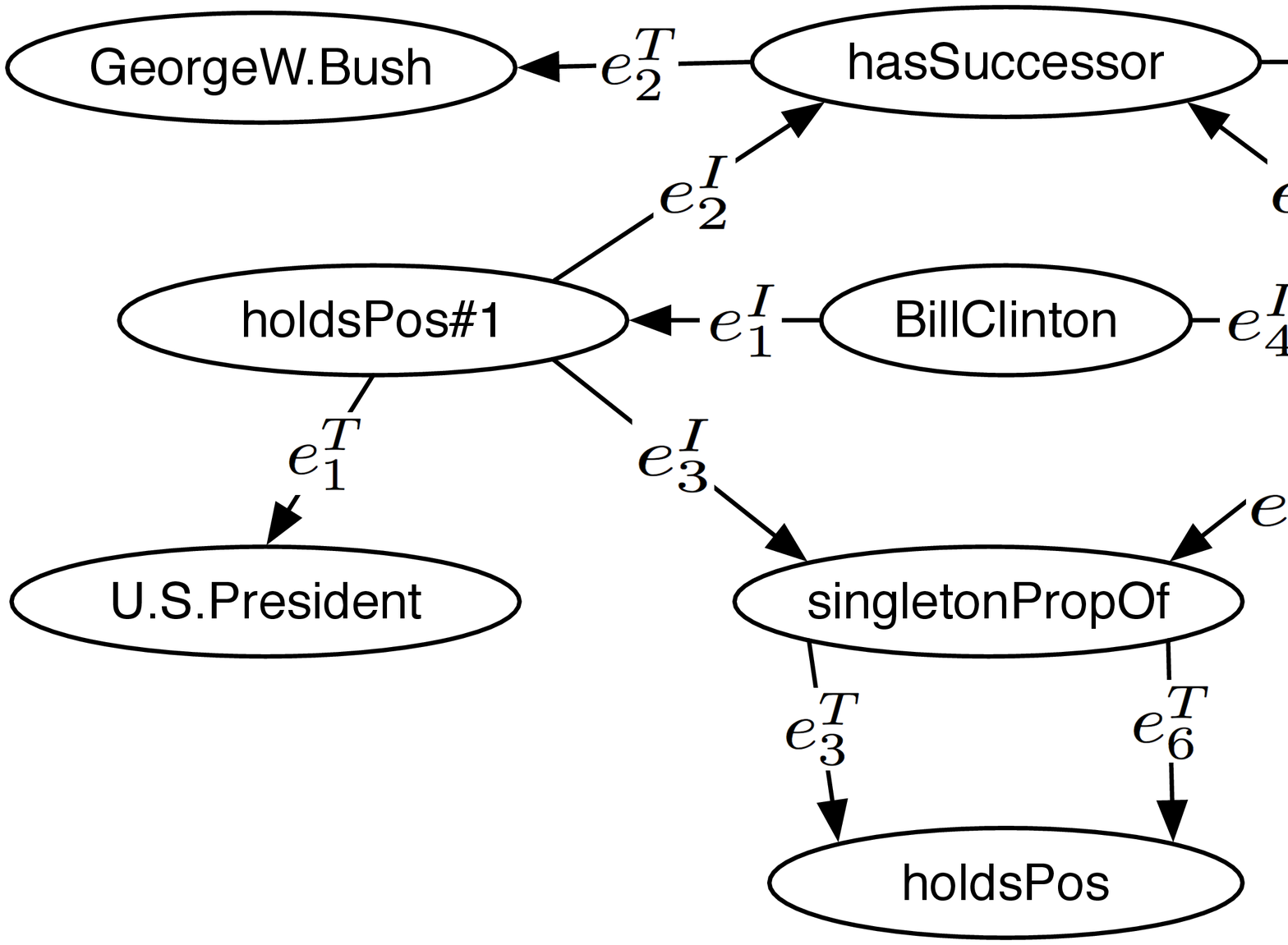}
 \caption{Labeled Directed Multigraph with Triple Nodes (LDM-3N) for the example in Table \ref{tbl-motivating-example}. \label{fg-motivating-example}}
\end{figure}

Compared to the NLAN model, the LDM-3N graph model provides better connectivity and makes it possible to find the answers for the reachability and shortest path queries. While the node \verb|GeorgeWBush| is unreachable from the node \verb|BillClinton| in the NLAN model, it is not only reachable but also neatly presented as the resource path (\verb|BillClinton|, \verb|holdsPos#1|, \verb|hasSuccessor|, \verb|GeorgeWBush|) in the LDM-3N model. We will formally define resource path and triple path in Section \ref{traversing-graph}.


In terms of practical impact, we show that compared to the NLAN model, our LDM-3N model does not introduce cost to the system in terms of space even though it conceptually adds two edges to connect (subject-predicate) and (predicate-object) in each triple. These initial and terminal edges that allow us to disambiguate the nodes forming each triple in the abstract model are unnecessary in the physical model. It is because the triples are already stored separately in the indices. We use the same indices for traversing in both models. In terms of execution time, traversing the NLAN graph is supposed to be faster than traversing the LDM-3N graph because exploration of the predicate nodes is bypassed. However, that extra exploration in the LDM-3N graph allows for additional paths that are impossible to find in the NLAN graph. We will report the extra time taken in traversing the LDM-3N graph through the shortest path algoritm in Section \ref{evaluation}.


We summarize our contributions in this paper next: 
\begin{itemize}
\item We develop a formal graph model (LDM-3N) for RDF. For traversing the new graph, we define two types of paths, resource path and triple path with a new traversal algorithm (in Section \ref{ldm-model}).
\item We rewrite the model-theoretic formal semantics for RDF with three levels of interpretation (simple, RDF, and RDFS) based on the LDM-3N model (Section \ref{semantics}) and demonstrate the graph entailments using RDFS deduction rules (Section \ref{entailments}), which the NLAN and BI models do not support.
\item We implement a new engine called GraphKE and extend the RDF-3X to support the LDM-3N graph model (Section \ref{implementation}). We implement and evaluate the empirical performance of the reachability and shortest path queries on both engines (Section \ref{evaluation}) to demonstrate that our formal model is technically viable for both new and existing systems.
\end{itemize}
We discuss related work and future work in Section \ref{related-work} and Section \ref{future}, respectively, and finally conclude the paper with Section \ref{conclusion}. 

\section{RDF Formal Graph Model}
\label{ldm-model}

We demonstrate how to represent any RDF triples in the LDM-3N model. We start with the step-by-step modeling of simple facts in Section \ref{simple-facts}, then formalize the LDM-3N model in Section \ref{formal-model}.

\subsection{Representing RDF Graph }
\label{simple-facts}

Referring back to the motivating example in Section \ref{motivating}, the NLAN diagram represents the triple \\
\indent \indent $T_3:$ \verb| holdsPos#1| \verb|hasSuccessor| \verb| GeorgeWBush| \\by mapping the predicate \verb|hasSuccessor| to an arc which connects the two nodes \verb|holdsPos#1| and \verb|GeorgeWBush|. We argue that, just as subjects and objects are mapped to nodes, predicates should also be mapped to nodes since they are all instances of the same class \verb|rdf:Resource|. The next question is ``How do we link those three nodes to form a triple?" 

Existing models use either labeled edges (NLAN) or one auxiliary node plus three labeled edges (BI). We will approach the question from another perspective. Although the predicate is mapped to a node, this predicate node differs from subject or object nodes because of its linking role. We need to represent this link without causing a discrepancy like the one found in the NLAN model. 

Here we have three nodes to form a triple. However, every edge in the graph can only connect two nodes. Instead of adopting a complicated solution like using a hyperedge, how about using two regular edges? Obviously, two directed edges can connect the three nodes into two pairs as illustrated in Fig. \ref{individual-example} (a). However, if we add another triple, say $T_6$: \verb|holdsPos#2 hasSuccessor FrankWhite|, then an ambiguity arises in figuring out which two edges form the triple as in Fig. \ref{individual-example} (b). The next question is ``How do we tie these two edges together to represent that they are part of the same triple?"

\begin{figure}[h!]
 \centering
 \includegraphics[width=0.45\textwidth]{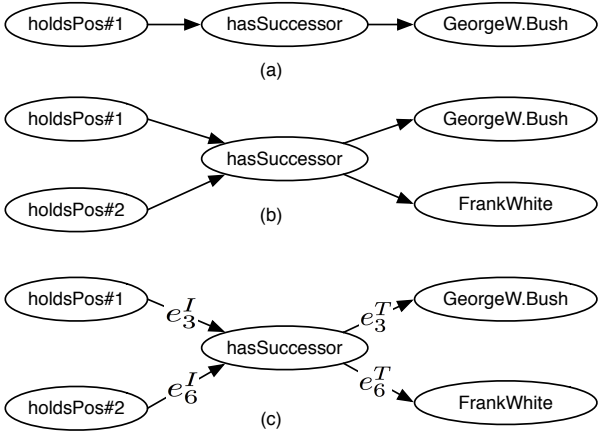}
 \caption{A step by step construction of a LDM-3N graph for the two triples $T_3$ and $T_6$. \label{individual-example}}
\end{figure}

Since we have two edges, we call the first edge between subject and predicate the initial edge, and we call the second edge between predicate and object the terminal edge. We create a mapping that maps every initial edge to its terminal edge. The pair of edges for a given triple $i$ is denoted as ($e^I_i, e^T_i$), where $e^I_i$ is the initial edge and $e^T_i$ is the terminal edge. It follows that ($e^I_3$, $e^T_3$) corresponds to $T_3$ and ($e^I_6$, $e^T_6$) to $T_6$. The complete LDM-3N graph for $T_3$ and $T_6$ is illustrated in Fig. \ref{individual-example} (c).

Next, we will formalize the LDM-3N model for any set of RDF triples.

\subsection{Formalizing RDF Graph}
\label{formal-model}
Formally, any RDF dataset is a set of RDF triples. Let $\mathcal{T}$ = \{$t_0, t_1, ..., t_n$\} be a set of triples on the vocabulary V, and let $G_{RDF}$ be the labeled directed multigraph with triple nodes (LDM-3N) of $\mathcal{T}$. We will first create the graph $G_{RDF}$ from $\mathcal{T}$ and then regenerate the set of original triples $\mathcal{T}$ from $G_{RDF}$.
\begin{proposition} \label{proposition-forward}
(Forward transformation). Any set of RDF triples can be transformed into a labeled directed multigraph with triple nodes $G_{RDF}$.
\end{proposition}

\begin{proof}
Let $V$ be the set of RDF terms in $\mathcal{T}$. Let $N$ and $E$ be the set of nodes and the set of directed edges in the graph $G_{RDF}$, respectively.

The bijective function $\mu: V \rightarrow N$ maps an RDF term in $V$ to a node in $N$.
Let $t_i$ be a triple in $\mathcal{T}$, $t_i = (s_i, p_i, o_i) \in \mathcal{T}$ with $0 \leq i \leq n$. Let $N_i \subset N$ such that $N_i = \{n_{si}, n_{pi}, n_{oi} | n_{si} = \mu(s_i), n_{pi} = \mu(p_i), n_{oi} = \mu(o_i)\}$, then $N = \bigcup_{i=0}^n N_i$.

The function $\epsilon : $ {\it E} $ \rightarrow$ {\it N} $\times$ {\it N} is defined to map every edge in $E$ to an ordered-pair of nodes. Let $e^I_i, e^T_i \in E: \epsilon(e^I_i) = (n_{si}, n_{pi})$ and $\epsilon(e^T_i) = (n_{pi}, n_{oi})$. 

The bijective function $\tau: $ {\it E} $ \rightarrow$ {\it E} maps an initial edge to a terminal edge of the same triple. Then $\tau(e^I_i) = e^T_i$. Let $E_i \subset E$ be the set of two edges representing $t_i$, $E_i = \{e^I_i, e^T_i\}$, and $E = \bigcup_{i=0}^n E_i$. 

Therefore, $G_i = (N_i, E_i, \epsilon, \tau, \mu)$ is the labeled directed graph with triple nodes of the triple $t_i$. 

Finally, with $N = \bigcup_{i=0}^n N_i$ and $E = \bigcup_{i=0}^n E_i$, the graph $G_{RDF} = (N, E, \epsilon, \tau, \mu)$ is a labeled directed multigraph with triple nodes for all of the triples in $\mathcal{T}$.
\qed
\end{proof}

\begin{proposition}(Backward transformation).
Given the graph $G_{RDF}(N, E, \epsilon, \tau, \mu)$ transformed by Proposition \ref{proposition-forward}, a set of RDF triples can be derived from $G_{RDF}$.
\end{proposition}

\begin{proof}

Let $e^I_i$ be any edge in $E$ with $0 \leq i \leq n$, the corresponding terminal edge of $e^I_i$ in $E$ is $e^T_i = \tau(e^I_i)$. From this pair of edges $e^I_i$ and $e^T_i$, we find the nodes connected by the two edges by using the $\epsilon$ function. Let $n_{si}, n_{pi}, n_{oi} \in N$ such that $(n_{si}, n_{pi}) = \epsilon(e^I_i)$, $(n_{pi}, n_{oi}) = \epsilon(e^T_i)$. 

Let $\mu^{-1}$ be the reverse function of $\mu$, then $\mu^{-1}: N \rightarrow V$ returns the RDF term mapped to a graph node. Let $s_i, p_i, o_i \in V$ such that $s_i = \mu^{-1}(n_{si})$, $p_i = \mu^{-1}(n_{pi})$ and $o_i = \mu^{-1}(n_{oi})$.
 The three nodes form the original triple $t_i = (s_i, p_i, o_i)$.
The set $\mathcal{T}$ of all RDF triples $t_i$ derived from $G_{RDF}$ is as follows:\\
$\mathcal{T} = \{t_i|\forall i: t_i = (s_i, p_i, o_i)$\}. 
\qed
\end{proof}

\begin{proposition} The size of the graph $G_{RDF}$ on the set of triples $\mathcal{T}$ is $2\vert\mathcal{T}\vert$.
\end{proposition}

\begin{proof}
The size of the graph $G_{RDF}$ is the number of edges needed to form all of the triples in $\mathcal{T}$. For each triple t = (s, p, o) $\in \mathcal{T}$, we need 2 edges to form the triple. Therefore, as the number of triples in $\mathcal{T}$ is $\vert\mathcal{T}\vert$, we need $2\vert\mathcal{T}\vert$ edges.
\qed
\end{proof}

\subsection{Traversing RDF Graph}
\label{traversing-graph}
A triple path is defined as a sequence of triple subgraphs where two adjacent triple subgraphs share one common node, that the subject node of the later triple subgraph is also the object or predicate node of the previous triple subgraph. In the given LDM-3N graph $G_{RDF}$, the triple path $tp$ = ($t_i$, $t_{i+1}$, ..., $t_j$) where $t_k$ = ($n_{sk}$, $n_{pk}$, $n_{ok}$) with $i \leq k \leq j$ and $t_{k+1}$ = ($n_{s(k+1)}$, $n_{p(k+1)}$, $n_{o(k+1)}$) satisfy one of following conditions: (1) $n_{s(k+1)}$ = $n_{pk}$ or (2) $n_{s(k+1)}$ = $n_{ok}$. 

For example, ($T_1$, $T_3$) from Table \ref{tbl-motivating-example} is a triple path because two triples share the node \verb|holdsPos#1|.

A resource path is defined as a sequence of nodes such that (1) every two adjacent nodes are connected by an edge, (2) every three nodes connected by a pair of initial and terminal edges should form a triple. In the given LDM-3N graph $G_{RDF}$, the resource path ($n_0$, $n_1$, ..., $n_k$) satisfies the following conditions:
\begin{enumerate}
\item ($n_i$, $n_{i+1}$) = $\epsilon(e_i)$ $\in$ $E$ for all $0 \leq i \leq k-1$ 
\item for every edge $e^T_i$ = ($n_i$, $n_{i+1}$) $\in$ $E$ and $\exists e^I_i \in E$ such that $\tau(e^I_{i}) = e^T_i$, 

if $\exists e^I_{i-1}$, $e^T_{i-1} \in E$ such that $e^I_{i-1}$ = ($n_{i-1}$, $n_{i}$) $\in$ $E$ and $\tau(e^I_{i-1}) = e^T_{i-1}$, then $\tau(e^I_{i-1}) = e^T_{i}$, or $e^T_{i-1} = e^T_{i}$.
\end{enumerate}

The second condition excludes the case where the two initial and terminal edges from different triples are adjacent to each other in a path. From the motivating example, the path (\verb|BillClinton|, \verb|holdsPos#1|, \verb|hasSuccessor|, \verb| GeorgeWBush|) satisfies these conditions and hence, it is a resource path. We will describe the implementation of the shortest resource path in Section \ref{shortest-resource-path} and its evaluation in Section \ref{evaluation}.

\section{Formal Semantics}
\label{semantics}
This section explains how a labeled directed multigraph with triple nodes, or LDM-3N graph, can be exploited as an underlying model for RDF(S) formal semantics. Similar to the model-theoretic semantics described in \cite{hitzler2011foundations,Nguyen:2014:DLR:2566486.2567973}, we also represent three levels of interpretation: simple, RDF, and RDFS. However, we define new functions for modeling the underlying labeled directed multigraphs and use them to redefine class/property extension functions. 

\subsection{Simple interpretation}
The {\bf simple interpretation} $\mathcal{I}$ of the vocabulary V based on the LDM model consists of: 
\begin{itemize}
\item {\it IN}, a non-empty set of {\it nodes} in $\mathcal{I}$,
\item {\it IE}, a set of {\it directed edges} in $\mathcal{I}$,

\item $\mathcal{I_E}:$ {\it IE} $\rightarrow$ {\it IN} $\times$ {\it IN}, mapping each edge to an ordered pair of nodes. This function may be many-to-one because a multigraph allows multiple edges between two nodes.
\item $\mathcal{I_T}:$ {\it IE} $\rightarrow$ {\it IE}, mapping an initial edge to a terminal edge of the same triple. This is a bijective function, and that the reverse function $\mathcal{I}^{- 1}_T:$ {\it IE} $\rightarrow$ {\it IE} maps a terminal edge to its initial edge.

\item {\it EL}, a set of distinct labels to be assigned to the edges in {\it IE},
\item $I_{EL}: EL \rightarrow IE$, a labeling function, mapping labels from {\it EL} into the set {\it IE} of edges. All labeling functions are also bijective.
\item $I_S: V \rightarrow IN$, a labeling function, mapping URIs from $V$ into {\it IN},
\item $LV$, a set of {\it literal values} in IN,
\item $I_L: V \rightarrow IN$, a labeling function, mapping literal values from $V$ to {\it IN}.
\end{itemize}
Let $\cdot^\mathcal{I}$ be the interpretation function that maps all the URIs and literals in $V$ to the set of nodes in {\it IN}. A ground triple ($s$ $p$ $o$ $ .)^\mathcal{I} $ is assigned true if all $s, p, o \in V$ and $\exists e_1, e_2 \in$ {\it IE} $:\mathcal{I_E}(e_1) = (s^\mathcal{I}, p^\mathcal{I}), \mathcal{I_E}(e_2) = (p^\mathcal{I}, o^\mathcal{I})$, and $\mathcal{I_T}(e_1) = e_2$.

The simple interpretation in \cite{hitzler2011foundations} contains an extension function that maps a property to a set of resource pairs. In this simple interpretation, we define a set of new functions, $\mathcal{I_E}$ and $\mathcal{I_T}$, for representing nodes and edges of the underlying labeled directed multigraphs with triple nodes.
To make this interpretation compatible with the existing ones, we incorporate the existing criteria for generic properties and singleton properties into this interpretation. The simple interpretation $\mathcal{I}$ also satisfies the following criteria:

\begin{itemize}
\item {\it IP}, a set of {\it generic property nodes}, which is also a subset of {\it IN}, {\it IP} $\subset$ {\it IN}.

\item $I_{EXT}$, a function assigning to each property node a set of pairs from {\it IN}.\\
 $I_{EXT}$ : {\it IP} $\rightarrow 2^{\it IN \times IN}$ where $I_{EXT}(p)$ is the {\it extension} of generic property {\it p}.
 Particularly, $I_{EXT}(p) = \{(s, o)| \exists e_1, e_2 \in$ {\it IE} $:\mathcal{I_E}(e_1) = (s, p), \mathcal{I_E}(e_2) = (p, o)$, and $\mathcal{I_T}(e_1) = e_2$\}.
 
 \item {\it IPs}, called the set of {\it singleton property nodes} of $\mathcal{I}$, as a subset of {\it IN}, 

 \item $I_{S\_EXT}(p_s)$, a function mapping a singleton property to a pair of resources.
$I_{S\_EXT}$ : {\it IPs} $\rightarrow {\it IN \times IN}$. Particularly,
 $I_{S\_EXT}(p_s) = (s, o)$ such that $\exists e_1, e_2 \in$ {\it IE} $:\mathcal{I_E}(e_1) = (s, p_s), \mathcal{I_E}(e_2) = (p_s, o)$, and $\mathcal{I_T}(e_1) = e_2$.

\end{itemize}

\subsection{RDF interpretation}
The {\bf RDF interpretation} of a vocabulary V is a simple interpretation $\mathcal{I}$ of the vocabulary V $\cup$ $V_{RDF}$ that satisfies the following criteria:
\begin{itemize}

\item {\it $p$} $\in$ {\it IP} if $\exists e_1, e_2 \in$ {\it IE} $:\mathcal{I_E}(e_1) = (p, $ rdf:type$^\mathcal{I}), \\ \mathcal{I_E}(e_2) = ($rdf:type$^\mathcal{I}, $ rdf:Property$^\mathcal{I})$,\\ and $\mathcal{I_T}(e_1) = e_2$. 
A generic property is an instance of the class rdf:Property.

\item {\it $p_s$} $\in$ {\it IPs} if $\exists e_1, e_2 \in$ {\it IE} $:\mathcal{I_E}(e_1) = (p_s, $ rdf:type$^\mathcal{I}), \\\mathcal{I_E}(e_2) = ($rdf:type$^\mathcal{I}, $ rdf:SingletonProperty$^\mathcal{I})$, \\and $\mathcal{I_T}(e_1) = e_2$. Every singleton property is an instance of the class rdf:SingletonProperty.

\item {\it $p_s$} $\in$ {\it IPs} if $\exists e_1, e_2 \in$ {\it IE}, $p \in $ {\it IP} $:\mathcal{I_E}(e_1) = (p_s, $ rdf:singletonPropertyOf$^\mathcal{I}), \\ \mathcal{I_E}(e_2) = ($rdf:singletonPropertyOf$^\mathcal{I}, p)$, and $\mathcal{I_T}(e_1) = e_2$. A singleton property is connected to a generic property via the rdf:singletonPropertyOf.

\item if {\it $p_s$} $\in$ {\it IPs} then $\exists! (e_1, e_2): \mathcal{I_E}(e_1) = (s, p_s), \mathcal{I_E}(e_2) = (p_s, o), \mathcal{I_T}(e_1) = e_2$, with $s, o \in$ {\it IN} and $e_1, e_2 \in$ {\it IE}. This ensures only one occurrence of a singleton property as a predicate of a triple.

\item if $``s"\string^\string^$rdf:XMLLiteral is in $V$ and s is a well-typed XML Literal, then $I_L(``s"\string^\string^$rdf:XMLLiteral$) \in LV$ and $\exists e_1, e_2 \in $ {\it IE}$ : \mathcal{I_E}(e_1) = (s,$ rdf:type$^\mathcal{I}),\\ \mathcal{I_E}(e_2) = ($rdf:type$^\mathcal{I}, $ rdf:XMLLiteral$^\mathcal{I}$), and $\mathcal{I_T}(e_1) = e_2$. All well-typed XML literals are in $LV$ and they also are instances of class rdf:XMLLiteral.

\item if $``s"\string^\string^$rdf:XMLLiteral is in $V$ and s is an ill-typed XML Literal, then \\ $I_L(``s"\string^\string^$rdf:XMLLiteral$) \notin LV$ and $\not\exists e_1, e_2 \in $ {\it IE}$ : \mathcal{I_E}(e_1) = (s,$ rdf:type$^\mathcal{I}),\\ \mathcal{I_E}(e_2) = ($rdf:type$^\mathcal{I}, $ rdf:XMLLiteral$^\mathcal{I}$), and $\mathcal{I_T}(e_1) = e_2$. \\
$LV$ does not contain ill-typed XML literal values.

\end{itemize} 

\subsection{RDFS interpretation}
In the RDFS interpretation, we define the function $I_{\it CEXT}: {\it IN} \rightarrow 2^{\it IN}$ 
where $I_{\it CEXT}(y)$ is called the class extension of $y$. Although some critiera can alternatively be expressed using these two functions $I_{CEXT}$ and $I_{EXT}$, here we explicitly express these criteria as graph constraints to demonstrate their graph nature. We show only some of the more important conditions here due to space constraints. 

{\bf RDFS interpretation} of a vocabulary V is an RDF interpretation $\mathcal{I}$ of the vocabulary V $ \cup$ $V_{RDF}$ $\cup$ $V_{RDFS}$ that satisfies the following criteria:

\begin{itemize}
\item $I_{CEXT}$ : {\it IP} $\rightarrow 2^{\it IN}$ , a function assigning to each class a set of nodes from {\it IN}. 
$I_{CEXT}(c)$ is called the {\it class extension} of class {\it c}.
 Particularly, $I_{CEXT}(c) = \{s | s \in $ {\it IN,} $\exists e_1, e_2 \in$ {\it IE} $:\mathcal{I_E}(e_1) = (s, $ rdf:type$^\mathcal{I}), \mathcal{I_E}(e_2) = ($rdf:type$^\mathcal{I}, c)$, and $\mathcal{I_T}(e_1) = e_2$\}.

\item if $\exists e_1, e_2, e_3, e_4 \in$ {\it IE}$ : \mathcal{I_T}(e_1) = e_2, \mathcal{I_T}(e_3) = e_4, \\\mathcal{I_E}(e_1) = (x, $ rdfs:domain$^\mathcal{I}), \mathcal{I_E}(e_2) = ($rdfs:domain$^\mathcal{I}, y)$, $ \mathcal{I_E}(e_3) = (u, x), \mathcal{I_E}(e_4) = (x, v)$, \\then $\exists e_5, e_6 \in$ {\it IE}$:\mathcal{I_T}(e_5) = e_6, \\\mathcal{I_E}(e_5) = (u, $ rdf:type$^\mathcal{I}), \mathcal{I_E}(e_6) = ($rdf:type$^\mathcal{I}, y)$.\\If one class is a domain of a property, then the class extension includes all subjects in the same triples with the property.

\item if $\exists e_1, e_2, e_3, e_4 \in$ {\it IE}$ : \mathcal{I_T}(e_1) = e_2, \mathcal{I_T}(e_3) = e_4, \\\mathcal{I_E}(e_1) = (x, $ rdfs:range$^\mathcal{I}), \mathcal{I_E}(e_2) = ($rdfs:range$^\mathcal{I}, y)$, $\mathcal{I_E}(e_3) = (u, x), \mathcal{I_E}(e_4) = (x, v)$, \\then $\exists e_5, e_6 \in$ {\it IE}$:\mathcal{I_T}(e_5) = e_6, \\\mathcal{I_E}(e_5) = (v, $ rdf:type$^\mathcal{I}), \mathcal{I_E}(e_6) = ($rdf:type$^\mathcal{I}, y)$. \\The class range of a property includes all objects in the same triples with the property.

\item if $\exists e_1, e_2 \in$ {\it IE}$ : \mathcal{I_T}(e_1) = e_2,\\ \mathcal{I_E}(e_1) = (x, $ rdfs:subPropertyOf$^\mathcal{I}),$ \\and $ \mathcal{I_E}(e_2) = ($rdfs:subPropertyOf$^\mathcal{I}, y)$, then $x, y \in ${\it IP} and $I_{EXT}(x) \subseteq I_{EXT}(y)$. The extension of a property is a subset of the extension of its super property.

\item if $\exists e_1, e_2 \in$ {\it IE}$ : \mathcal{I_T}(e_1) = e_2,\\ \mathcal{I_E}(e_1) = (x, $ rdfs:subClassOf$^\mathcal{I}),$ \\and $ \mathcal{I_E}(e_2) = ($rdfs:subClassOf$^\mathcal{I}, y)$, \\then $I_{CEXT}(x) \subseteq I_{CEXT}(y)$. The extension of a class is a subset of its super class extension.

\end{itemize}

\subsection{RDFS Entailments}
\label{entailments}
Here we present how the LDM-3N graph-based semantics described in Section \ref{semantics} can derive other graphs using RDFS rules. We demonstrate the entailments by using three rules: rdfs5, rdfs7, and rdfs9 from \cite{hitzler2011foundations}. We choose these rules since they are commonly used for reasoning with class and property hierarchy. Other rules can also be applied in the same way.
Let $G$ and $G'$ be the two LDM-3N graphs.

\begin{itemize}
\item rdfs5: if ($u$, rdfs:subPropertyOf, $v$) and \\($v$, rdfs:subPropertyOf, $x$) \\then ($u$, rdfs:subPropertyOf, $x$). This rule states that the rdfs:subPropertyOf property is transitive.\\
if $\exists e_1, e_2, e_3, e_4 \in$ {\it IE} in $G : \mathcal{I_T}(e_1) = e_2, \mathcal{I_T}(e_3) = e_4, \\\mathcal{I_E}(e_1) = (u, $ rdfs:subPropertyOf$^\mathcal{I}), \\\mathcal{I_E}(e_2) = ($rdfs:subPropertyOf$^\mathcal{I}, v),\\ \mathcal{I_E}(e_3) = (v, $ rdfs:subPropertyOf$^\mathcal{I}), \\\mathcal{I_E}(e_4) = ($rdfs:subPropertyOf$^\mathcal{I}, x) $, \\then
$\exists e_5, e_6 \notin$ {\it IE} $ : \mathcal{I_T}(e_5) = e_6, \\\mathcal{I_E}(e_5) = (u, $ rdfs:subPropertyOf$^\mathcal{I}),\\ \mathcal{I_E}(e_6)$ $= ($rdfs:subPropertyOf$^\mathcal{I}, x)$. The graph $G$ entails $G'$: $IE' = IE \cup \{e_5, e_6\}$. 

\item rdfs7: if ($a$, rdfs:subPropertyOf, $b$) and ($u$, $a$, $y$) then ($u$, $b$, $y$).\\ All resources interlinked by a property are interlinked by its super property.\\
if $\exists e_1, e_2, e_3, e_4 \in$ {\it IE} in $G : \mathcal{I_T}(e_1) = e_2, \\\mathcal{I_E}(e_1) = (a, $ rdfs:subPropertyOf$^\mathcal{I}),\\ \mathcal{I_E}(e_2) = ($rdfs:subPropertyOf$^\mathcal{I}, b),\\ \mathcal{I_E}(e_3) = (u, a), \mathcal{I_E}(e_4) = (a, y) $ and $ \mathcal{I_T}(e_3) = e_4$, then 
$\exists e_5, e_6 \notin$ {\it IE} $ : \mathcal{I_T}(e_5) = e_6, \mathcal{I_E}(e_5) = (u, b), \mathcal{I_E}(e_6) = (b, y)$, and the graph $G$ entails $G': IE' = IE \cup \{e_5, e_6\}$. 

\item rdfs9: if ($v$, rdf:type, $u$) and ($u$, rdfs:subClassOf, $x$) then ($v$, rdf:type, $x$).\\ A resource as a member of one class is also a member of its superclass.\\
if $\exists e_1, e_2, e_3, e_4 \in$ {\it IE} in $G : \mathcal{I_T}(e_1) = e_2, \mathcal{I_T}(e_3) = e_4, \\\mathcal{I_E}(e_1) = (u, $ rdfs:subClassOf$^\mathcal{I}), \\\mathcal{I_E}(e_2) = ($rdfs:subClassOf$^\mathcal{I}, v), \\\mathcal{I_E}(e_3) = (v, $ rdfs:subClassOf$^\mathcal{I}), \\\mathcal{I_E}(e_4) = ($rdfs:subClassOf$^\mathcal{I}, x) $ \\ then
$\exists e_5, e_6 \notin$ {\it IE:} $\mathcal{I_T}(e_5) = e_6, \\\mathcal{I_E}(e_5) = (u, $ rdfs:subClassOf$^\mathcal{I}), \\\mathcal{I_E}(e_6) = ($rdfs:subClassOf$^\mathcal{I}, x)$,\\ and the graph $G$ entails $G': IE' = IE \cup \{e_5, e_6\}$. 
\end{itemize}

\section{Implementation}
\label{implementation}
Here we describe the implementation of the two engines in order to demonstrate the feasibility of practical application of our graph model. Since we are proposing a new graph model with a new way of traversing the graph, our aim is not to optimize the shortest path algorithm because there is a multitude of existing solutions for this purpose, such as \cite{Angles:2008:SGD:1322432.1322433}. Instead, our aim is to implement the two traversal algorithms from two graph models on the same system so that we can fairly evaluate the two graph models. We also explain how existing systems can adopt this model by making changes to their traversal algorithm. 

\subsection{Shortest Resource Path}
\label{shortest-resource-path}
\begin{algorithm}
\caption{Dijkstra's algorithm on the LDM-3N graph}\label{alg:dijkstraldm}
\begin{algorithmic}[1]
\Procedure{LDM-3N\_Dijkstra}{$source, target$}
\State $distance \gets 0$
\State $pqueue \gets source$
\While{$(curid \gets pqueue\_pop(pqueue, dis))\not=0$}
 \If{(curid = target)} 
\State{\textbf{return} dis} 
\Else
\State{get all the pairs <pred, obj> of curid} 
\For{each pair <pred, obj>} 
\State{$ret \gets update\_node(pred, dis+1, curid)$}
 \If{(ret = 1)} 
\State{pqueue\_push(pqueue, pred, dis + 1)}
\EndIf
\State{\textbf{end}}
\State{$ret \gets update\_node(obj, dis+2, pred)$}
 \If{(ret = 1)} 
\State{pqueue\_push(pqueue, obj, dis + 2)} 
\EndIf
\State{\textbf{end}}
\EndFor
\State{\textbf{end}}
\EndIf
\State{\textbf{end}}
\EndWhile\label{euclidendwhile}
\State{\textbf{end}}
\EndProcedure
\State{\textbf{end}}
\end{algorithmic}
\end{algorithm}

We implemented Dijkstra's algorithm on both NLAN and LDM-3N graph models to find the shortest resource path as defined in Section \ref{traversing-graph}. One difference between the two algorithms is that traversing the LDM-3N graph visits the predicates and explores their neighbor nodes while traversing the NLAN does not.

While traversing the graph, we use an in-memory priority queue to store all the nodes to be visited and their distance to the source node. We also use a hash table \verb|visited| to store the set of visited nodes. For each visited node, we keep track of the previous node-id and the shortest distance to the source node. Here we explain our traversal algorithms on both graph models.

Algorithm \ref{alg:dijkstraldm} starts with the source pushed into the priority queue \verb|pqueue|. At each step, one specific node \verb|curid| with its distance \verb|dis| to the source is removed from the priority queue and explored. All the pairs <pred, obj> are looked up from the hash index. Each \verb|pred| (with distance \verb|dis + 1| and \verb|curid| as previous node) or \verb|obj| (with distance \verb|dis + 2| and \verb|pred| as previous node) will be updated in the hash table \verb|visited|. If the \verb|pred| or \verb|obj| hasn't been visited before, or its new distance is shorter than the one in the hash table, the hash table \verb|visited| will be updated with the new distance for that node, and the \verb|curid| will also be pushed into the \verb|pqueue|.  This process is repeated until the target is found or the \verb|pqueue| becomes empty, which means no path exists between the source and the target.

\begin{algorithm}
\caption{Dijkstra's algorithm on the NLAN graph}\label{alg:dijkstranlan}
\begin{algorithmic}[1]
\Procedure{NLAN\_Dijkstra}{$source, target$}
\State $distance \gets 0$
\State $pqueue \gets source$
\While{$(curid \gets pqueue\_pop(pqueue, dis))\not=0$}
 \If{(curid = target)} 
\State{\textbf{return} dis} 
\Else
\State{get all the pairs <pred, obj> of curid} 
\For{each pair <pred, obj>} 
\State{$ret \gets update\_node(obj, dis + 1, curid)$}
 \If{(ret = 1)} 
\State{pqueue\_push(pqueue, obj, dis + 1)} 
\EndIf
\State{\textbf{end}}
\EndFor
\State{\textbf{end}}
\EndIf
\State{\textbf{end}}
\EndWhile\label{euclidendwhile}
\State{\textbf{end}}
\EndProcedure
\State{\textbf{end}}
\end{algorithmic}
\end{algorithm}

Dijkstra's algorithm \ref{alg:dijkstranlan} for the NLAN model differs from the Dijkstra's algorithm \ref{alg:dijkstraldm} for the LDM-3N model in that it does not visit the predicate as explained from lines 10-13 of Algorithm \ref{alg:dijkstraldm}. Instead, it only traverses from the \verb|curid| to the \verb|obj| and the distance \verb|dis + 1|. We can observe that the distance associated with the \verb|obj| in this model is shorter than the one associated with the \verb|obj| in the LDM-3N model. That is because the NLAN model traverses from \verb|sub| to \verb|obj| in one hop, while the LDM-3N model traverses from \verb|sub| to \verb|obj| in two hops, with \verb|pred| in the middle.

\subsection{GraphKE}

We implemented GraphKE on top of Berkeley DB (BDB) key-value store in C language \cite{olson1999berkeley}. This implementation can also be adapted to other key-value stores. 

\textbf{Dictionaries}. All the URIs and literals are mapped to internal identifiers using 8 bytes. We use odd numbers for literal identifiers and even numbers for the rest. As a literal cannot be subject of any triple, it will never get explored when traversing the RDF graph. We mark them as a sink node with odd numbers. We leave the URIs and the literals as they are. We did not compress them although compressing the strings may help save space. We created two dictionaries using either the hash index or the B-tree index supported by BDB. From our initial evaluation we observed that the hash index performed better than the B-tree index. However, this may not be the case for other datasets. Therefore, we do not fix the use of the hash index for all datasets. The type of index to be used can be specified before loading data.

\textbf{Data triples}. Each triple is internally represented in the form of <subject-id, predicate-id, object-id>. We loaded the triples into a hash index with subject-id as the key and a pair of <predicate-id, object-id> as the value. We also created an extra index with subject-id as the key and the number of <predicate-id, object-id> pairs as the value. 

\textbf{Basic graph operator}. In order to support the LDM-3N graph traversal, given a node-id, we need to lookup for all of the adjacent nodes. We use a database cursor to iterate through the data items to find all the pairs (predicate-id, object-id).





\section{Evaluation}
\label{evaluation}
This section describes the empirical evaluation on our graph models through their shortest path algorithms described in Section \ref{shortest-resource-path}.

We deployed the two engines GraphKE and RDF-3X into the same server running Ubuntu 12.04.4 LTS with 256 GB of RAM and 220 GB of SSD hard drive. Since RDF-3X is self-tuned, no configuration is necessary. For Berkeley DB, we set the cache size to 64GB in the configuration.

\subsection{Dataset}
We use the YAGO2S-SP\footnote{http://wiki.knoesis.org/index.php/Singleton\_Property} dataset generated by Nguyen et al. \cite{Nguyen:2014:DLR:2566486.2567973}. This dataset was chosen because it is suitable for our purpose with 62 million singleton properties representing the facts that are attached with different kinds of metadata. 
We load this dataset into the two engines, excluding the file WikipediaInfo.ttl. This file contains triples with the \verb|linksTo| predicate which does not provide a meaningful relationship between the resources. After removing this file, the dataset contains 267,161,278 triples with 77,895,604 URI nodes and 31,110,161 literals. We loaded this dataset into both GraphKE and RDF-3X. 

In order to evaluate our graph traversal algorithms, we need to run them with varying lengths. We are able to find resource paths of distances up to 139 in our experiments. Here we explain how we randomly generated different sets of <source, target> input pairs for the path algorithms from querying the Yago2S-SP dataset. 

We observe that the singleton properties of the property \verb|holdsPoliticalPosition| and the property \verb|hasSuccessor| form very long paths. The basic pattern is similar to the ($T_1$, $T_2$, $T_3$) of the motivating example, but the pattern is repeated multiple times and forms a much longer path. 

\begin{lstlisting}[captionpos=b,label=SPARQL query for generating input,
   basicstyle=\ttfamily \small]
BASE <http://yago-knowledge.org/resource/>
SELECT ?politician
WHERE {
 ?politician ?sp1 ?position .
 ?sp1 rdf:singletonPropertyOf 
 		<holdsPoliticalPosition> . 
}
GROUP BY ?position
\end{lstlisting}

We run the SPARQL query above to find the set of politicians and we group these politicians by their political positions. We then choose three political positions that have a good number of politicians: \\(1) White House Chief of Staff (CS),\\(2) Secretary of State (SS), and \\(3) Speakers of the U.S. House of Representatives (HR). For each group of politicians, we generate a set of <politician1, politician2> and <politician2, politician1> pairs, and we use three sets for the evaluation of the reachability and shortest path queries.

\begin{table}
\centering
\caption{Number of politicians and pairs for each input group
\label{tbl-input-groups}}
\begin{tabular}{clll } \noalign{\smallskip}\hline \noalign{\smallskip}
 \textbf{$Group$} & \textbf{$N_o$ of politicians} & \textbf{$N_o$ of pairs} \\ \hline \noalign{\smallskip}
 \bf{HR} & 51 & 2550 \\ 
 \bf{SS} & 34 & 1122 \\ 
 \bf{CS} & 22 & 462 \\ 
 \hline \noalign{\smallskip}
 \end{tabular}
\end{table}

\subsection{Reachability Queries}

For reachability queries, we adjusted Dijkstra's algorithms described in Section \ref{shortest-resource-path}. For every pair <source, target>, if the algorithm starts with the source and finishes exploring all the nodes in the priority queue without reaching the target, we report that the target is not reachable from the source.

Given the three sets of input pairs generated in the manner explained above, we run both NLAN and LDM-3N algorithms on the GraphKE to find the set of pairs that are reachable. Since we are dealing with thousands of pairs, we expedite the process by running the algorithms on multiple threads. Running the LDM-3N with 1 thread took 9,848 seconds (2.73 hrs) to finish processing 1,122 input pairs in the group CS. When we ran it with 5 threads, it only took 3,230 seconds (54 minutes) and averaged 7.2 seconds per pair. We report the experiment with 5 threads in Table \ref{tbl-reachable-groups}. Two observations are noted from this evaluation.

First, by traversing the LDM-3N graph one can find hundreds of reachable pairs, while one traversing the NLAN graph does not find any from all three input groups. For example, out of 462 pairs from 22 politicians holding the position of White House Chief of Staff, 164 pairs are reachable in the LDM-3N. Overall, out of 4,134 pairs from 3 position groups, 837 pairs are reachable.

Second, the algorithm running on the LDM-3N takes about 6 seconds on average, while the one running on the NLAN takes about 2 seconds. Running on the NLAN graph is 3 times faster than the algorithm running on the LDM-3N graph. This result is straightforward and confirms our hypothesis that the NLAN should be faster because it only traverses a subset of nodes that are visited by the LDM-3N. It does not traverse to the predicates and their neighbors. As a consequence, NLAN misses all of the paths connected through singleton properties as predicates, which is a big loss since this dataset contains about 67 million singleton properties.

\begin{table}
\centering
\caption{Number of reachable pairs (R) and time taken (T) per group in the LDM-3N model and the NLAN model, ran in 5 threads in GraphKE \label{tbl-reachable-groups}}
\begin{tabular}{clllllll } \noalign{\smallskip}\hline \noalign{\smallskip}
 \multicolumn{2}{c|}{} & \multicolumn{3}{c|}{\textbf{LDM-3N}} & \multicolumn{3}{c|}{\textbf{NLAN}}\\  \noalign{\smallskip}
 & \textbf{Pairs} & \textbf{R} & \textbf{T(s)} & \textbf{Avg} & \textbf{R} & \textbf{T(s)} & \textbf{Avg}\\ \hline \noalign{\smallskip}
 \bf{HR} & 2550 & 579 & 14782 & 5.797 & 0 & 5592 & 2.19 \\ 
 \bf{SS} &1122 & 94  & 7721 & 6.881 & 0 & 2419 & 2.155 \\ 
 \bf{CS}  &462 & 164  & 3230 & 7.186 & 0 & 885 & 1.915 \\ 
  \hline \noalign{\smallskip}
\bf{All}  &4134 & 837  & 25733 & \bf{6.23} & 0 & 8896 & \bf{2.15} \\ 
 \hline \noalign{\smallskip}
 \end{tabular}
\end{table}

\subsection{Shortest Resource Path Queries}

Getting the results from the reachability queries described above, we collected one subset of reachable pairs from each input group and used them as the input for the shortest path algorithms. We ran the shortest path algorithms on both GraphKE and RDF-3X in single-threaded mode. For each input pair, we printed out the shortest resource path and its associated triple path. We ran each set of reachable pairs three times on GraphKE and two times on RDF-3X, and calculated the average of these runs. Afterwards, since each reachable input group contains a number of paths sharing the same distance, we also obtained the average time for each shortest distance within each group, and reported them in Figures \ref{shortest-results-cs}, \ref{shortest-results-ss}, and \ref{shortest-results-hr} for three groups CS, SS, and HR, respectively. Table \ref{tbl-multithread-groups} summarizes the total time and average time taken in seconds in the LDM-3N model running in 1 thread and 5 threads in GraphKE, and 1 thread in RDF-3X.

\begin{figure}[h!]
 \centering
 \includegraphics[width=0.48\textwidth,natwidth=983,natheight=514]{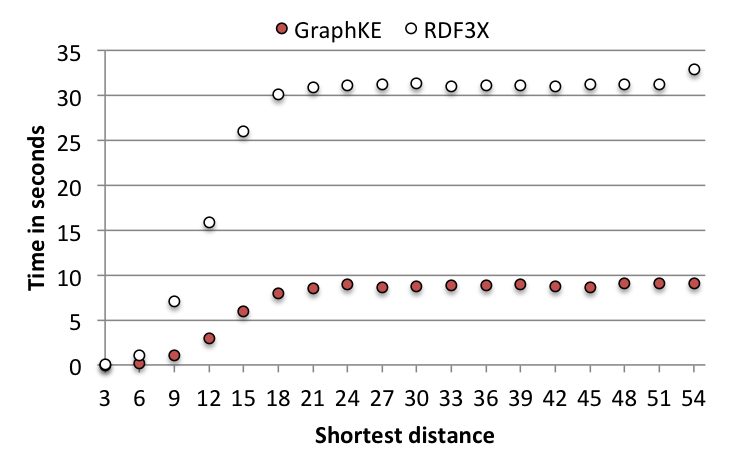}
 \caption{Average time taken per shortest distance within CS. \label{shortest-results-cs}}
\end{figure}
\begin{figure}[h!]
 \centering
 \includegraphics[width=0.48\textwidth,natwidth=983,natheight=514]{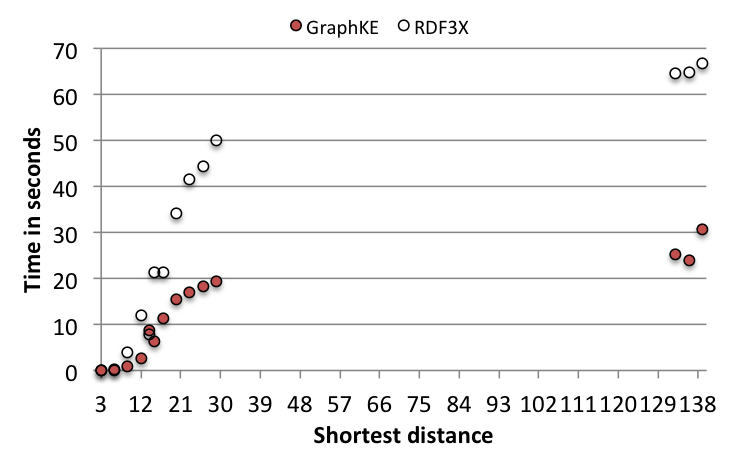}
 \caption{Average time taken per shortest distance within SS. \label{shortest-results-ss}}
\end{figure}
\begin{figure}[h!]
 \centering
 \includegraphics[width=0.48\textwidth,natwidth=752,natheight=453]{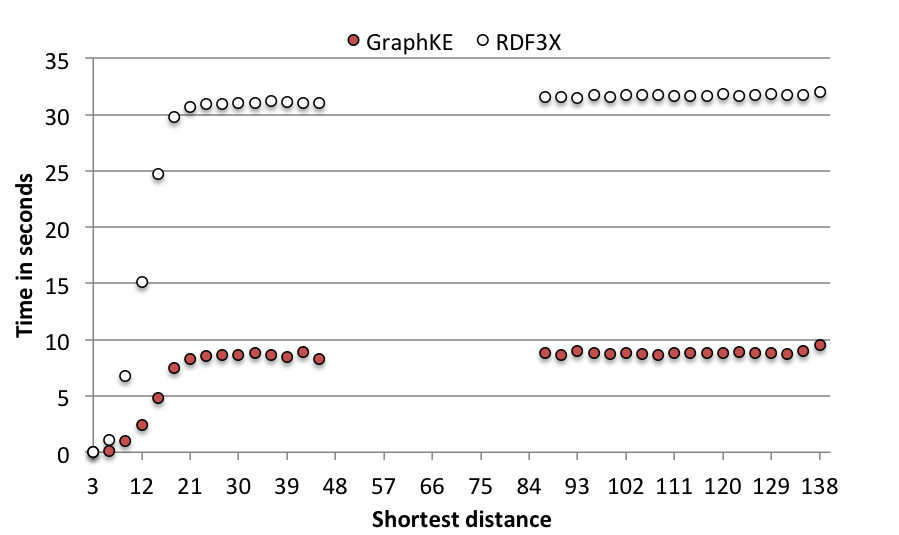}
 \caption{Average time taken per shortest distance within HR. \label{shortest-results-hr}}
\end{figure}

Regarding our choice of reporting on RDF-3X in addition to GraphKE, we want to clarify that our focus is on demonstrating the feasibility of meeting practically useful application needs by extending an existing system, as well as developing a new system based on the new model we defined. Although we use the same plots to show the results for two alternatives, the purpose is not to directly compare the two.

\textbf{GraphKE.} Our new implementation on GraphKE runs with 1 thread in 353 sec (6 mins), 770 sec (13 mins), 2,688 (45 mins) for 94 (SS), 164 (CS), and 579 (HR) inputs, respectively. It takes less than 10 seconds for every input pair in CS, HR, and for input pairs with shortest distances up to 15 in SS. For the shortest distances from 15 to 139 in SS, it takes up to 30 seconds. On average, it takes 4.553 seconds for an input with 1 thread as shown in Table \ref{tbl-multithread-groups}.

When we run the GraphKE with 5 threads, it finishes processing all input pairs in 2,582 seconds (44 minutes), compared to 3,811 seconds (64 minutes) in 1 thread. We observe that the time taken for each input pair increases in 5 threads, but the overall time taken for all input pairs reduces 30\% compared to 1 thread.
\begin{table}
\centering
\caption{Total time taken and average time taken in seconds in the LDM-3N model running in 1 thread vs. 5 threads in GraphKE, and 1 thread in RDF-3X \label{tbl-multithread-groups}}
\begin{tabular}{clllllll } \noalign{\smallskip}\hline \noalign{\smallskip}
 \multicolumn{2}{c|}{} & \multicolumn{4}{c|}{\textbf{GraphKE}} & \multicolumn{2}{c|}{\textbf{RDF-3X}}\\  \noalign{\smallskip}
 \multicolumn{2}{r}{Pairs} & \textbf{1T} & \textbf{$Avg_1$} & \textbf{5T} & \textbf{$Avg_5$} & \textbf{1T} & \textbf{$Avg_1$}\\ \hline \noalign{\smallskip}
 \bf{HR} & 579 & 2688 & 4.64 & 1826& 3.15&10627 & 18.35\\ 
 \bf{SS} & 94 & 353 & 3.76 & 237 &  2.52 & 888 & 9.45 \\ 
 \bf{CS}  & 164 & 770 & 4.695 & 519 & 3.17 & 2388 & 14.56\\ 
 \hline \noalign{\smallskip}
 \bf{All}  & 837 & 3811 & \bf{4.55} & 2582 & \bf{3.08} & 13903 & \bf{16.61}\\ 
 \hline \noalign{\smallskip}
 \end{tabular}
\end{table}



\textbf{Findings.} While investigating the results from our algorithm, we found many paths that are interesting to us. We draw one sample path in Figure \ref{interesting_paths}. This LDM-3N graph illustrates the use of singleton properties in representing the n-ary relationship between the politician, the position, and the successor. The politicians are connected to their position and successor of that position by a set of singleton properties represented by the rounded squares.
\begin{figure}[h!]
 \centering
 \includegraphics[width=0.50\textwidth,natwidth=983,natheight=514]{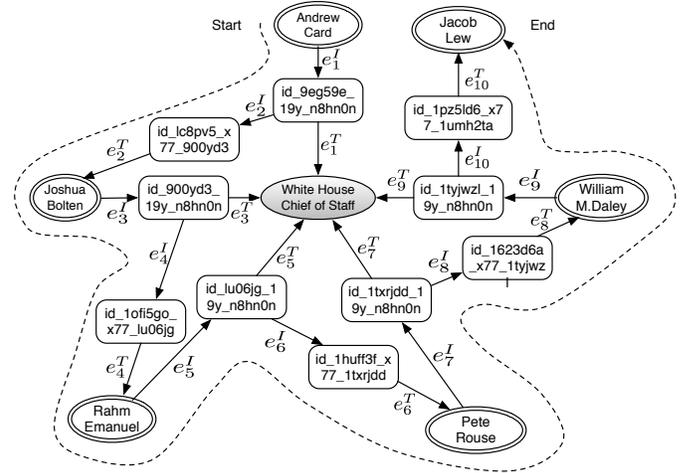}
 \caption{The shortest resource path and its respective triple path from Andrew Card to Jacob Lew in Yago2S-SP. The resource path includes all the edges along the dashed line in the periphery of the figure, and the triple path includes all the triples where three nodes from a triple form a straight line. Traversing the resource path from start to end will find a sequence of five politicians (Andrew Card, Joshua Bolten, Rahm Emanuel, Pete Rouse, William M. Daley, Jacob Lew) in consecutive terms.\label{interesting_paths}}
\end{figure}

This graph results from running Dijkstra's algorithm on the LDM-3N graph model to find the shortest path from Andrew Card to Jacob Lew. The shortest distance is 15. Starting from Andrew Card, we follow the outer edges along the dashed line in 15 hops, we will find a sequence of politicians Joshua Bolten, Rahm Emanuel, Pete Rouse, William M. Daley and then we will reach Jacob Lew. For every two politicians in the sequence, the former is succeeded by the latter.


\subsection{Discussion}

We defined a new formal graph model for RDF that does not suffer from limitations of current alternative. To demonstrate viability for practical implementation and use, we showed two implementations - GraphKE as a ground up implementation and adaptation of RDF-3X, for executing reachability and shortest path queries. Based on the experimental results in the Yago2S-SP dataset with 260 million triples, we believe that the LDM-3N graph model can be effectively implemented in a practical system, either new or existing.

\section{Related Work}
\label{related-work}
Directly related to the formal graph model for RDF are the NLAN diagram, which is currently used in the W3C Recommendation documents, and the bipartite model. As explained in prior sections, our approach differs from the NLAN approach in that predicates are mapped to nodes instead of arcs, and the nodes within the same triple are interconnected by a pair of initial and terminal edges. Our model differs from the bipartite model in that we use a pair of directed edges, whereas the bipartite model uses one extra node and three extra edges to connect this node to subject, predicate, and object.

READ MORE ABOUT CONCEPTUAL GRAPH \footnote{http://www.w3.org/DesignIssues/CG.html}

In a broader context, a comprehensive collection of graph database models are described in the survey paper by Angles and Gutierrez \cite{Angles:2008:SGD:1322432.1322433}. The formal foundation of these graph models varies based on the basic definition of a mathematical graph, such as directed or undirected, labeled or unlabeled, graph or hypergraph, and node or hypernode. These models differ from our model in that they all represent predicates as labeled arcs, while we map them to nodes. From this perspective, the intuition of our approach is similar to that of conceptual graphs (CGs) \cite{sowa1998conceptual} which also models relations as nodes. Our approach, however, differs from the CGs in the mechanism to form a statement or formula. We map one relation to a single node and add constraints to bind every two edges in order to form a statement. CGs allow the same relation to be mapped to multiple nodes when the same relation occurs in multiple formulas. 

\section{Ongoing and Future Work}

\label{future}
\textbf{Ongoing.} We are participating on the ongoing community activities initiated by National Center for Biotechnology Information (NCBI) and Air Force Research Laboratory (AFRL) on applying the singleton property approach in creating RDF knowledge bases in the biomedical and material sciences.

\textbf{Future work.} A document of the LDM-3N model that fully supports the latest specifications of RDF syntax, semantics, and deduction rules is necessary (as we skipped some due to space constraints).
We believe that many applications in the Semantic Web areas such as graph databases, knowledge representation, graph theory, and logics could benefit from our LDM-3N graph model.

Adopting this LDM-3N model and the singleton property approach would allow more well-connected knowledge graphs that are aware of temporal, spatial, and provenance to be represented in RDF. For example, it is directly applicable to publishing, querying, and browsing the Linked Open Data (LOD). 
Well-connected knowledge graphs would be amenable to graph-based algorithms, which could be applied to both RDF data and schema triples. Moreover, standard algorithms well-studied from graph theory could be directly applied to these well-formalized graphs, while also leveraging the RDF(S) semantics. We are also interested in studying the possibility of performing graph entailments directly on the real graph structures. 

\section{Conclusion}

\label{conclusion}
We have presented a new formal graph model for RDF with examples clearly demonstrated. Our LDM-3N model allows us to represent any set of RDF triples in a formal graph. It also allows us to develop an underlying graph model for the model-theoretic semantics of RDF(S) and the RDFS entailments. To the best of our knowledge, this is the first formal graph model that is compatible with the RDF formal semantics. We have implemented the LDM-3N graph model in the new GraphKE engine and in the existing triple store RDF-3X. We have evaluated the empirical aspects of our graph model in the Yago2S-SP dataset to demonstrate its practical feasibility for handling real-world applications. 



\bibliographystyle{abbrv}
\bibliography{all,knowledge}
\additionalauthors{Author 4, Author 5}

\end{document}